\documentclass[11pt]{article}   
\usepackage{amsthm,amsmath}    
\usepackage[dvips]{color}
\usepackage{amssymb}
\usepackage{latexsym}
\usepackage{hyperref}

\theoremstyle{plain}
\newtheorem{theorem}{Theorem}[section]
\newtheorem{corollary}[theorem]{Corollary}

\theoremstyle{definition}
\newtheorem{definition}[theorem]{Definition}
\newtheorem{example}[theorem]{Example}

\def\keywords{%
\list{}{\advance\topsep by0.35cm\relax\small\rm
 \leftmargin=1cm
 \itemindent\listparindent
 \rightmargin\leftmargin}\item[\hskip\labelsep\bf Keywords: ]}

\theoremstyle{remark}

\renewenvironment{proof}{\noindent{\it Proof}.}{\qed}

\title{On the Complexity of Coordinated Table Selective Substitution Systems}  

\author{Liliana Cojocaru\\
Faculty of Computer Science\\ 
Alexandru Ioan Cuza University\\ 
liliana.cojocaru@info.uaic.ro} 
\date{}

\begin{document}
\thispagestyle{empty}
\maketitle
\thispagestyle{empty}
\pagenumbering{arabic}

\begin{abstract}
We investigate computational resources used by {\it Turing machines} (TMs) and {\it alternating 
Turing machines} (ATMs) to accept languages generated by {\it coordinated table selective 
substitution systems} with two components. We prove that the class of languages generated by 
{\it real-time ($RL$; $0S$)-systems}, an alternative device to generate {\it $\lambda$-free labeled 
marked Petri net} languages, can be accepted by nondeterministic TMs in ${\cal O}(\log n)$ space 
and ${\cal O}(n\log n)$ time. Consequently, this proper subclass of Petri nets languages (known 
also as ${\cal L}$-languages) is included in $NSPACE$($\log n$). The class of languages 
generated by {\it ($RL$; $RB$)-systems} for which the nonterminal alphabet of the $RL$-grammar is 
composed of only one symbol and the nonterminal alphabet of the $RB$-grammar is composed of 
two symbols, can be accepted by ATMs in ${\cal O}(\log n)$ time and space. Consequently, this proper 
subclass of {\it one counter languages} generated by {\it one counter machines with only one control 
state} is included in $U_{E^*}$-uniform ${\cal N}{\cal C}^1$,  hence in $SPACE$($\log n$).    

\end{abstract}
\vspace*{-0.7cm}

\section{Introduction}
{\it Coordinated table selective substitution systems} ({\it cts systems}) were 
introduced in \cite{R2} as an alternative framework to simulate grammars and automata. At 
the basis of cts systems stands the notion of {\it selective substitution grammar} ({\it s-grammar}) 
which is a unifying model for several rewriting systems and grammars \cite{R1}, \cite{R2}. The main 
idea in an s-grammar is to substitute various kinds of rules by only context-free (CF) rules, that  
at a certain step of derivation must rewrite a certain combination of symbols occurring in the sentential 
form. This combination is provided by the so called {\it selector} of the s-grammar, built in such a way 
it simulates the respective rewriting formalism. Informally, an s-grammar $G$ is composed of a set of CF 
productions $P$ and a selector $K$, which is a language over an alphabet $\Sigma\cup \bar\Sigma$, where 
$\Sigma$ is the alphabet of $G$ and $\bar\Sigma=\{\bar a| a\in \Sigma\}$, $\Sigma\cap \bar \Sigma=\emptyset$. A word $y\in \Sigma^*$ 
is derived from $x\in \Sigma^*$, if and only if there exists $z\in K \subseteq (\Sigma\cup \bar \Sigma)^*$ such that $x=z$, after replacing 
all letters in $\bar\Sigma$ occurring in $z$ by their unbarred counterparts in $\Sigma$, and each barred letter in $z$ is rewritten in the 
usual manner by rules in $P$. This is a derivation step in an s-grammar $G$. The language of $G$ is the set of all words derived from a fixed 
symbol $S$, called the {\it axiom} of $G$, by iteratively applying the derivation step described above. If $K\subseteq \Sigma^*\bar\Sigma\Sigma^*$, 
then $K$ is called  {\it sequential}. The sequential selectors we are concerned with are the {\it right-boundary selector} $K=\Sigma^*\bar \Sigma$ 
and the {\it 0-sequential selector} $K=\Sigma^*\bar\Sigma \Sigma^*$. 

 A cts system \cite{AR} consists of $n$ s-grammars  $G_1,...,G_n$ and a set $R\subseteq P_1\times...\times P_n$ of rewrites. Each $P_i$,  $1\hspace*{-0.1cm}\leq i\hspace*{-0.1cm}\leq n$,  is the set of productions of $G_i$. 
An $n$-tuple $(y_1,...,y_n)$ can be derived from $(x_1,...,x_n)$, where $x_i$ and $y_i$ are words 
over the alphabet of $G_i$, if there exists a rewrite $r\hspace*{-0.1cm}=(r_1,...,r_n)\hspace*{-0.1cm}\in R$ such that each $y_i$ can be derived from $x_i$ in $G_i$ using $r_i$. The language of the cts system consists 
of all words generated on the first component, such that the system starts the computation with a fixed $n$-tuple 
of words called the {\it system axiom}, and ends up the computation in an $n$-tuple composed of empty words 
excepting the first component (containing the system's word). Cts systems with one component ($n=1$) simulate 
the work of grammars, while cts systems with $n\geq 2$ components simulate automata machines. 
Cts systems with two coordinates that uses a right-boundary selector, denoted by ($RL$; $RB$)-systems, 
are an alternative model for pushdown automata. Cts systems with two coordinates using a $0$-sequential 
selector, denoted by ($RL$; $0S$)-systems, are an alternative model for generating Petri nets (PN) languages. 
The computational power of cts systems and their relation to classes of languages in the Chomsky hierarchy 
have been studied in \cite{EHR2},  \cite{KR}, and \cite{R2}. However, in none 
of these papers the authors investigate the computational complexity of these systems, since this may 
have interesting consequences on the complexity of some subclasses of languages in the Chomsky 
hierarchy (e.g. one counter languages). In this paper we try to cover this gap.  

We investigate computational resources used by \textit{(alternating) Turing machines} 
to accept languages generated by ($RL$; $0S$) and ($RL$; $RB$)-systems. We prove in Section 3 that 
real-time ($RL$; $0S$)-languages can be accepted by  nondeterministic on-line TMs in 
${\cal O}(\log n)$ space. Consequently, the class of languages generated by 
$\lambda$-free labeled marked PN (known as ${\cal L}$-languages \cite{H}) is included in $NSPACE$($\log n$). 
This shows the existence of a substantial gap between the general ($EXPSPACE$-hard) PN languages, known as 
${\cal L}^\lambda_0$-languages \cite{H}, and ${\cal L}$-languages. The class of \textit{one counter languages} 
can be generated by ($RL$; $RB$)-systems for which the nonterminal alphabet of the $RB$-grammar is composed 
of only two symbols, the axiom and a working symbol. Denote by ($RL$; $RB_c$) these systems. 
In Section 4  we provide an alternative proof, via ($RL$; $RB_c$)-systems, of the inclusion 
of one-counter languages in $NSPACE$($\log n$). If we further restrict the $RL$-grammar to be 
composed of only one nonterminal then the resulting ($RL$; $RB$)-systems are equivalent with \textit{one counter} machines with \textit{one control state}. We prove that the class of languages generated by one state one counter machines can be accepted by \textit{indexing} ATMs in ${\cal O}(\log n)$ time and space. Since $ALOGTIME$, the class of languages recognizable by indexing ATMs in logarithmic time, equals $U_{E^*}$-uniform $\cal N$$\cal C$$^1$ \cite{BDG}, \cite{R} we conclude that this proper subclass of one counter languages belongs to \textit{$\cal N$$\cal C$$^1$}, hence to $SPACE$($\log n$). 
In Section 5 we provide several examples of languages generated by ($RL$; $RB_c$)-systems.

\section{Prerequisites}  

We introduce the main concepts related to ($RL$; $0S$) and ($RL$; $RB$)-systems. We assume the 
reader to be familiar with basic notions of automata and formal language theory \cite{Sa}. For 
an alphabet $\Sigma$, $\Sigma^*$ denotes the free monoid generated by $\Sigma$, $\lambda$ is the 
empty string, $|x|_a$ denotes the number of occurrences of the letter $a$ in the string 
$x\in \Sigma^*$, $|\Sigma|$ is the cardinality of $\Sigma$, while $|x|$ the length of 
$x\in \Sigma^*$. We denote by $id_\Sigma$ the homomorphism from $(\Sigma\cup \bar \Sigma)^*$ into 
$\Sigma^*$, defined by $id_\Sigma(\bar a)=a$ and $id_\Sigma(a)=a$, for all $a\in \Sigma$. From 
\cite{AR} and \cite{R2} we have  

\begin{definition}\rm
(1) A selector over $\Sigma$ is a subset of $(\Sigma\cup \bar\Sigma)^*\bar\Sigma(\Sigma\cup \bar\Sigma)^*$.\\ 
\hspace*{0.2cm}(2) A table is a triple $T\hspace*{-0.1cm}=(\Sigma, P, K)$, where $\Sigma$ is the alphabet of $T$, $P\subseteq \Sigma\times\Sigma^*$ 
is a finite nonempty set of productions, and $K$ is a selector over $\Sigma$.  \\ 
\hspace*{0.2cm}(3) Let $T\hspace*{-0.1cm}=(\Sigma, P, K)$ be a table. For $x, y\in \Sigma^*$ we say that $x$ directly derives $y$ in $T$, denoted by 
$x\Rightarrow_{T}y$, if $x=x_1...x_n$, $y=y_1...y_n$, $n\geq 1$, $x_i\in\Sigma$, $y_i\in\Sigma^*$, $1\leq i\leq n$, and there exists $z\in K$, 
$z=z_1...z_n$, $z_i\in\Sigma\cup\bar \Sigma$ such that $id_\Sigma(z)=x$ and, for $1\leq i\leq n$, if 
$z_i\in \Sigma$, then $x_i=y_i$, and if $z_i\in \bar\Sigma$, then $(x_i,y_i)\in P$.
\end{definition}

\begin{definition}\rm
\hspace*{-0.1cm}(1) A {\it right-linear grammar} (denoted $RL$-grammar) is a $5$-tuple $G=(\Sigma, P, S, \Delta, K)$, where  
$T=(\Sigma, P, K)$ is the table of $G$, $\Sigma=al(G)$ is the alphabet of $G$, $\Delta=tal(G)\subseteq \Sigma$ is the terminal alphabet of $G$, $\Sigma-\Delta=ntal(G)$ is the nonterminal alphabet of $G$, $S=ax(G)\in ntal(G)$ is the axiom of $G$, $P=prod(G)$ is the set of productions in $T$ such that to any $(X, \alpha)\in prod(T)$ it corresponds a rule 
of the form $X\rightarrow \alpha$, $X\in ntal(G)$, $\alpha \in \Sigma\cup \{\lambda\} \cup (tal(G) ntal(G))$,
$K=(tal(G))^*\overline{(ntal(G))}$. \\
\hspace*{0.2cm}(2) A {\it $0$-sequential grammar} (denoted $0S$-grammar) or a {\it right-boundary grammar} (denoted $RB$-grammar) is a 
$4$-tuple $G=(\Sigma, P, S, K)$, where $T=(\Sigma, P, K)$ is the table of $G$, $\Sigma=al(G)=ntal(G)$ is the alphabet of $G$, $S=ax(G)\in \Sigma$ is the axiom of $G$, $P=prod(G)$ is the set of productions in $T$ such that to any $(X, Y)\in P$ it corresponds a rule of the form  
$X\rightarrow Y$, $X\in \Sigma$, $Y \in \Sigma^*$, $K=\Sigma^*\bar \Sigma\Sigma^*$  for a $0S$-grammar, and respectively $K=\Sigma^*\bar {\Sigma}$ for an $RB$-grammar.
\end{definition}

\begin{definition}\rm
\hspace*{-0.1cm}(1) A {\it right-linear $0$-sequential system}, denoted ($RL$; $0S$)-system, 
is a triple $G=(G_1, G_2,R)$, where $G_1$ is an $RL$-grammar, $G_2$ is an $0S$-grammar, and $R=rew(G)\subseteq prod(G_1)\times prod(G_2)$ is 
the set of rewrites of $G$.\\
\hspace*{0.2cm}(2) A {\it right-linear right-boundary system}, abbreviated ($RL$; $RB$)-system, 
is a triple $G=(G_1, G_2,R)$, where $G_1$ is an $RL$-grammar, $G_2$ is an $RB$-grammar, and $R=rew(G)\subseteq prod(G_1)\times prod(G_2)$ is the 
set of rewrites of $G$.  
\end{definition}

\begin{definition}\rm   
Let $G=(G_1, G_2, R)$ be an $X$-system, $X\in \{(RL; 0S),(RL; RB)\}$, and $x=(x_1,x_2)$, $y=(y_1,y_2)\in (al(G_1))
^*\times (al(G_2))^*$.\\ 
\hspace*{0.2cm}(1) We say that $x$ {\it directly derives} $y$ in $G$ by using $r=(r_1,r_2)\in rew(G)$, 
denoted by $x\hspace{-0.1cm}\Rightarrow^r_G y$, if  $x_1\hspace{-0.1cm}\Rightarrow^{r_1}_{G_1} \hspace{-0.1cm}y_1$ and 
$x_2\hspace{-0.1cm}\Rightarrow^{r_2}_{G_2} \hspace{-0.1cm}y_2$. Let $\Rightarrow_G^*$ be the reflexive and 
transitive closure of $\Rightarrow_G$. If $x \Rightarrow_G^* y$, then we say that $y$ is derived from $x$ in $G$.\\   
\hspace*{0.2cm}(2) Let $(x,y)\in (al(G_1))^+\times(al(G_2))^+$. An $(x,y)$-derivation in $G$ is a sequence $\rho=\rho(0), ..., \rho(n)$ of elements in  $(al(G_1))^*\times(al(G_2))^*$, $n\geq 0$, such that $\rho(0)=(x,y)$ and for $1\leq i\leq n$, 
$\rho(i-1) \Rightarrow_G  \rho(i)$. In the sequel  $\rho$ is called an $n$-step derivation and $\rho(i-1) \Rightarrow_G  \rho(i)$ is called the $i^{th}$ step derivation of $\rho$. Each  $\rho(i)$, $0\leq i\leq n$, 
is called a {\it snapshot} of $\rho$.\\ 
\hspace*{0.2cm}(3) Let $\rho=\rho(0), ..., \rho(n)$, $n\geq 0$, be a derivation in $G$. 
We say that $\rho$ is a successful derivation if $\rho(0)=(ax(G_1),ax(G_2))$ and $\rho(n)=(w, \lambda)$ where $w \in \Delta^*$.\\ 
\hspace*{0.2cm}(4) The language generated by $G$ is defined as  
$L(G)=\{w\in (tal(G_1))^*| (ax(G_1), ax(G_2))\Rightarrow^*_G (w,\lambda)\}$. $L(G)$ is referred to as an $X$-language, $X\in \{(RL; 0S),(RL; RB)\}$.       
\end{definition}


The class of all $X$-languages is denoted by ${\cal L}(X)$, $X\in \{(RL; 0S),(RL; RB)\}$.    
The strings $x_1,y_1 \in (al(G_1))^*$, $x_2,y_2 \in (al(G_2))^*$ (in Definition 2.4) are called sentential 
forms (s.forms) of the $RL$-, $0S$- and respectively, $RB$-grammars. 

An $RL$-grammar $G$ is {\it real-time} if $prod(G)$ contains no chain rules, i.e., rules of the form 
$X\rightarrow Y$, where $X, Y\in ntal(G_1)$. Let $G=(G_1, G_2, R)$ be an $X$-system. $G$ is {\it real-time} 
if $G_1$ is real-time. The class of real-time $X$-languages is denoted by ${\cal L}_{rt}(X)$, 
$X\in \{(RL; 0S),(RL; RB)\}$.       

Henceforth, in any reference to an $X$-system $G=(G_1, G_2, R)$, $X\in \{(RL; 0S),(RL; RB)\}$, $R=rew(G)=\{r_1,...,r_k\}$ 
is considered to be the ordered set of productions in $prod(G_1)\times prod(G_2)$, $A_1=ax(G_1)$, 
$B_1=ax(G_2)$, $ntal(G_1)=\{A_1,...,A_{l_1}\}$ and $al(G_2)=ntal(G_2)=\{B_1,...,B_{l_2}\}$ 
the ordered set of nonterminals in $G_1$ and $G_2$, respectively. Each $r_i\in rew(G)$, $1\leq i \leq k$, 
is a pair of the form $(r_{i,G_1}, r_{i,G_2})$, where $r_{i,G_1}\in prod(G_1)$ and $r_{i,G_2}\in prod(G_2)$. 
Each rule $r_{i,G_1}$ is of the form $\alpha_{r_{i,G_1}}\rightarrow \beta_{r_{i,G_1}}$, 
$\alpha_{r_{i,G_1}}\in ntal(G_1)$ and $\beta_{r_{i,G_1}}\in (tal(G_1)\cup \{\lambda\})(ntal(G_1)\cup \{\lambda\})$. 
Each $r_{i,G_2}$ is of the form $\alpha_{r_{i,G_2}}\rightarrow \beta_{r_{i,G_2}}$, 
$\alpha_{r_{i,G_2}}\in ntal(G_2)$ and $\beta_{r_{i,G_2}}\in (ntal(G_2))^*$.  
The \emph{net effect} of rule $r_{i,G_2}$, with respect to $B_l\in ntal(G_2)$, is provided by $df_{B_l}(r_{i,G_2})=|\beta_{r_{i,G_2}}|_{B_l}-|\alpha_{r_{i,G_2}}|_{B_l}$, $1\leq l \leq l_2$. To each rule $r_{i,G_2}\in prod(G_2)$ we associate a vector  $V(r_{i,G_2}) \in $ \textbf{Z}$^{l_2}$ 
defined by $V(r_{i,G_2})=(df_{B_1}(r_{i,G_2}),df_{B_2}(r_{i,G_2}),..., df_{B_{l_2}}(r_{i,G_2}))$. 


\section{On the Complexity of ($RL$; $0S$)-Systems and Petri Nets} 

{\it Petri Nets} (PN) \cite{H}, \cite{J}, \cite{P}, equivalent with {\it Vector Addition Systems with States}  (VASS) \cite{Ho},  are a language modeling  framework to simulate reactive and distributed systems, with practical applications in concurrent business processes and economy, in chemical or biological processes. Formally we have 

\begin{definition}\rm
A labeled marked Petri net, abbreviated lmPN, is a $6$-tuple $\cal P$ $=(P, T, F, \Sigma,\\ l, M_0)$, where $P$ is the set of places, $T$ is the set of transitions, $F \subseteq (P\times T)\cup (T\times P)$ is the flow relation, $\Sigma$ is the alphabet, $l$ is the labeling function $l:T\rightarrow \Sigma\cup \{\lambda\}$ that assigns to each transition either a letter or $\lambda$, and $M_0$ is the initial marking (a marking of a PN is a function from $P$ to $\textbf{N}$). If to a lmPN we associate a finite set of (final) markings then the rezulted PN is called a labeled marked Petri net with final markings, abbreviated lmPNf. 
If  $l(t)\neq \lambda$ for each $t\in T$, then  $\cal P$ is a $\lambda$-free lmPN(f).  
\end{definition}

The language generated by a $\lambda$-free lmPN $\cal P$ is the set of all labeled firing-sequences of $\cal P$ 
from $M_0$ to any marking. The corresponding class of languages is denoted in \cite{H} and \cite{J} by ${\cal L}$.  
The language generated by a lmPN $\cal P$ is the set of all labeled firing-sequences of $\cal P$ from $M_0$ to some final markings. The class of languages generated by lmPN  with final markings is denoted in \cite{H} by ${\cal L}^{\lambda}_0$. From \cite{H} we have ${\cal L}\subset {\cal L}^{\lambda}_0$. The membership problem for ${\cal L}^{\lambda}_0$ is many-one equivalent to the reachability problem for VASS \cite{H}, which was proved to be $EXPSPACE$-hard \cite{L}, as a lower bound. Nowadays this problem is known to have a non-elementary complexity \cite{S}, i.e., a 
tower of exponentials of time of height depending elementarily on the size of the input, as lower bound provided by $TOWER=\textbf{F}_3$ \cite{C}, and a non-primitive recursive upper bound provided by $ACKERMANN = \textbf{F}_{\omega}$ \cite{LS1}. On the other hand the reachability problem for 1-VASS (VASS in dimension\footnote{By dimension of a VASS it is understood the number of counters of VASS.} one) with unary encoding is 
$NSPACE(\log n)$-complete \cite{ELP}. If binary encoding is used, then the reachability problem for 1-VASS is $NP$-complete \cite{HK}. We prove that, in spite of the ${\cal L}^{\lambda}_0$'s non-elementary complexity, the membership problem of ${\cal L}$, in binary encoding, has a sub-linear space complexity, i.e., ${\cal L}\subseteq NSPACE(\log n)$. We have 

\begin{theorem}\rm
Each language $L\in$ ${\cal L}_{rt}$($RL$; $0S$) can be recognized by an on-line nondeterministic Turing machine 
in ${\cal O}(\log n)$ space and ${\cal O}(n\log n)$ time (i.e., ${\cal L}_{rt}$($RL$; $0S$) $\subseteq NSPACE$($\log n$)).  
\end{theorem}
\begin{proof}
Let $G=(G_1, G_2, R)$ be a real-time ($RL$; $0S$)-system and $L(G)$ the language generated by $G$. 
Let $\cal T$ be an on-line nondeterministic TM (with stationary positions) composed of an input tape 
that stores an input $w\in (tal(G_1))^*$, $w =a_1a_2...a_n$ of length $n$, and one (read-write) working tape. 

From an initial state $q_0$, in stationary positions, $\cal T$ records (by using constant time, related 
to the length of the input) the net effects of the rules in $prod(G_2)$, and it books $l_2$ Parikh blocks 
${\cal B}_l$, $1\leq l \leq l_2$, of ${\cal O}(\log n)$ cells (to record the Parikh vector of the  s.forms of $G_2$). 
At the beginning of the computation, on the first cell of each ${\cal B}_l$, $\cal T$ records 
the Parikh vector $V_{G_2}^0$ of the axiom $B_1=ax(G_2)$, i.e., $V^0_{G_2,B_1}=1$ and $V^0_{G_2,B_l}=0$, $2\leq l \leq l_2$. 
Denote by $q_s$ the state reached at the end of this procedure. In brief, from now on $\cal T$ proceeds as follows. 
For each scanned symbol $a_i$, $\cal T$ checks whether there exists at least one rule in $rew(G)$ such that the rule 
in $G_1$ rewrites the single nonterminal existing in the s.form of $G_1$, producing $a_i$, and whether the nonterminal in the left-hand side of the rule in $G_2$ exists in the s.form of $G_2$. $\cal T$ memorizes in states the information regarding the rules applied at this step and updates the Parikh vector of the s.form of $G_2$. It accepts the input if $\cal T$ reaches the end of $w$ with no nonterminals in the s.forms of $G$. 

Formally, $\cal T$  works as follows. From $q_s$, $\cal T$ enters in a state $q_{r^{a_1,X_1}_{\hspace*{-0.1cm}A_1}}$ by reading $a_1$, where $r^{a_1,X_1}_{A_1}=(r^{a_1,X_1}_{A_1,G_1}, r_{\hspace{-0.1cm}\_,G_2})$ and $r^{a_1,X_1}_{A_1,G_1}$ is a rule of the form $A_1\rightarrow a_1X_1$. As $X_1$ may vary, there may exist several states of this  type (the nondeterminism). If there is no rule $r^{a_1,X_1}_{A_1} \in rew(G)$ then the computation is blocked, hence $w$ is not accepted. On $q_{r^{a_1,X_1}_{A_1}}$ $\cal T$ checks (in stationary positions) whether $r_{\hspace{-0.1cm}\_,G_2}$ rewrites $B_1$. If this holds\footnote{Otherwise, i.e.,  $r_{\hspace{-0.1cm}\_,G_2}$ does not rewrite $B_1$, then the computation is blocked on this state, but may continue on another state of type $q_{r^{a_1,X_1}_{A_1}}$, for a different $X_1$.}, $\cal T$ adds the net effect of  $r_{\hspace{-0.1cm}\_,G_2}$, with respect to $B_l$, to $V^0_{G_2,B_l}$, $1\leq l \leq l_2$, i.e., $\cal T$ computes, in binary, $sdf^{(1,G_2)}_{B_l}=V^0_{G_2,B_l}+df_{B_l}(r_{\hspace{-0.1cm}\_,G_2})$. The binary value of 
$sdf^{(1,G_2)}_{B_l}$ is recorded on the $l^{th}$ Parikh block of $G_2$. From state $q_{r^{a_1,X_1}_{A_1}}$, 
$\cal T$ enters in a state of the form $q_{r^{a_2,X_2}_{X_1}}$ by reading $a_2$, where $r^{a_2,X_2}_{X_1}=(r^{a_2,X_2}_{X_1,G_1}, r_{\hspace{-0.1cm}\_,G_2})$, 
$r^{a_2,X_2}_{X_1,G_1}$ is of the form $X_1\rightarrow a_2X_2$, and $r_{\hspace{-0.1cm}\_,G_2}$ is the rule in $prod(G_2)$ such that $r^{a_2,X_2}_{X_1}\in rew(G)$. 
If there is no rule in $prod(G_1)$ that rewrites $X_1$,  then the 
computation is blocked. On $q_{r^{a_2,X_2}_{X_1}}$  $\cal T$ checks whether the nonterminal rewritten  by $r_{\hspace{-0.1cm}\_,G_2}$ exists in the s.form of $G_2$, i.e., whether $sdf^{(1,G_2)}_{\alpha_{r_{\hspace{-0.1cm}\_,G_2}}}\hspace{-0.1cm}>0$, where 
$\alpha_{r_{\hspace{-0.1cm}\_,G_2}}$ is the left-hand side of  $r_{\hspace{-0.1cm}\_,G_2}$. If it holds, $\cal T$ adds 
the net effect of $r_{\hspace{-0.1cm}\_,G_2}$ with respect to $B_l$, to $sdf^{(1,G_2)}_{B_l}$,  $1\leq l\leq l_2$, 
i.e., $\cal T$ computes, in binary,  $sdf^{(2,G_2)}_{B_l}=sdf^{(1,G_2)}_{B_l}\hspace{-0.1cm}+df_{B_l}(r_{\hspace{-0.1cm}\_,G_2})$. 
The space used by $\cal T$ to record $sdf^{(1,G_2)}_{B_l}$ is now used to record  $sdf^{(2,G_2)}_{B_l}$.

The computation continues in this way until all symbols $a_i$, $1\leq i\leq n$, are read and checked as 
explained above. Besides, for the last symbol $a_n$, $\cal T$ also checks whether 
the s.form of $G_2$ contains no nonterminal, i.e.,  $sdf^{(n,G_2)}_{B_l}=sdf^{(n-1,G_2)}_{B_{l}}+df_{B_l}(r_{\hspace{-0.1cm}\_,G_2})=0$, 
$1\leq l \leq l_2$. 
\vspace{0.1cm}

The input is accepted if for each  $a_i$, $1\leq i\leq n$, 
there exists at least one state of the form $q_{r^{a_i,X_i}_{X_{i-1}}}$, where 
$r^{a_i,X_i}_{X_{i-1}}=(r^{a_i,X_i}_{X_{i-1},G_1}, r_{\hspace{-0.1cm}\_,G_2})\in rew(G)$,  
$r^{a_i,X_i}_{X_{i-1},G_1}$ is of the form $X_{i-1}\rightarrow a_iX_i$, respectively $X_{n-1}\rightarrow a_n$, $r_{\hspace{-0.1cm}\_,G_2}$ rewrites an arbitrary occurrence of $\alpha_{r_{\hspace{-0.1cm}\_,G_2}}$\footnote{Recall that $\alpha_{r_{\hspace{-0.1cm}\_,G_2}}$ is the nonterminal placed on the left-hand side of rule 
$r_{\hspace{-0.1cm}\_,G_2}$.} in the s.form of $G_2$, and at the end of computation $\cal T$ reaches a final state with no nonterminal in the s.forms of $G$. 
Recall that $w \in L(G)$ if for each snapshot there exists a strategy of marking a symbol in the s.form of $G_2$ 
such that, rewriting this symbol, at the end of computation, the s.form of $G_2$ equals $\lambda$. Since $K=\Sigma^*\bar \Sigma\Sigma^*$, it is evident that the position of the marked nonterminal (hence, the order in which the nonterminals are rewritten) in the s.form of $G_2$ is not important. There is no context that may control the marked symbol. Hence, the condition for a rule $r_{\hspace{-0.1cm}\_,G_2}$, $2\leq i \leq n$, to rewrite an arbitrary occurrence of $\alpha_{r_{\hspace{-0.1cm}\_,G_2}}$ suffices, as long as this nonterminal exists in the s.form of $G_2$ (checked by the condition $sdf^{(i-1,G_2)}_{\alpha_{r_{\hspace{-0.1cm}\_,G_2}}}\hspace{-0.1cm}>0$). 

Since $G_2$ is a CF grammar, the length of each s.form in $G_2$ is linearly bounded by the length 
of the input. Hence, ${\cal O}(\log n)$ cells are enough in order to record the number of occurrences of each nonterminal $B_l$, $1\leq l \leq l_2$, on the s.form of $G_2$, and the space used by $\cal T$ is ${\cal O}(\log n)$. 
Because each time, reading an input symbol, $\cal T$ has to visit ${\cal O}(\log n)$ cells in the working tape, and 
the constant number of auxiliary (summation) operations with binary numbers, performed at each step of derivation in $G$, require ${\cal O}(\log n)$ time, $\cal T$ performs the whole computation in ${\cal O}(n\log n)$ time
Hence ${\cal L}_{rt}$($RL$; $0S$)$\subseteq NSPACE$($\log n$).
\end{proof}

In \cite{AR} it is proved that the class of languages generated by real-time ($RL$; $0S$)-systems is equal with the class of languages generated by $\lambda$-free labeled marked Petri nets, denoted in \cite{H} as ${\cal L}$.  Hence, from Theorem 3.2 we have

\begin{corollary}\rm 
${\cal L}={\cal L}_{rt}$($RL$; $0S$) ${\subseteq}$ $NSPACE(\log n)$.
\end{corollary}

\section{On the Complexity of ($RL$; $RB$)-Systems and One Counter Machines}

($RL$; $RB$)-systems are alternative models of {\it pushdown automata} \cite{EHR1}, \cite{EHR2}, \cite{EHR3}. Pushdown automata accept exactly the CF languages. If the pushdown alphabet is restricted to only two symbols, a symbol $S_2$ to mark the 
bottom of the stack and a working symbol $Z_2$, then the resulting machine is a {\it one counter automaton} (OCA). 
A counter can be increased by 1, decreased by 1, left unchanged, or tested for $0$  (the bottom marker test). The class of languages accepted by OCAs, denoted by ${\cal L}(OC)$, is a proper subset of CF languages \cite{Ha}. If OCAs are not allowed to perform $0$-tests then the resulting machines are called {\it one counter nets automata} (OCN). The class of languages accepted by OCN  are denoted by ${\cal L}(OCN)$. It is known that ${\cal L}(OCN)= {\cal L}(1$-$VASS)\subset {\cal L}(OC)$. We prove that the membership problem for some particular subclasses of ${\cal L}(OC)$ and ${\cal L}(OCN)$ are in ${\cal N}{\cal C}^1$.  

\begin{definition}\rm
(1) A cts-system $G=(G_1, G_2,R)$ is an ($RL$; $RB$)-counter system, denoted by ($RL$; $RB_c$), if $G_1$ is an $RL$-grammar, $G_2$ is an 
$RB$-grammar, $al(G_2)=ntal(G_2)=\{S_2,Z_2\}$, and $prod(G_2)\subseteq \{S_2\rightarrow S_2, S_2\rightarrow S_2Z_2, 
Z_2\rightarrow Z_2, Z_2\rightarrow Z_2Z_2, Z_2\rightarrow \lambda, S_2\rightarrow \lambda\}$. 

(2) A cts-system $G=(G_1, G_2,R)$ is an ($RL$; $RB$)-counter nets system, denoted by ($RL^0$; $RB_c$), if $G_1$ is an $RL$-grammar, $G_2$ is an 
$RB$-grammar, $al(G_2)=ntal(G_2)=\{Z_2\}$, and $prod(G_2)\subseteq \{Z_2\rightarrow Z_2, Z_2\rightarrow Z_2Z_2, Z_2\rightarrow \lambda\}$. 
\end{definition}

 From  \cite{R2} we have ${\cal L}$($RL$;$RB_c$) = ${\cal L}(OC)$ and ${\cal L}$($RL^0$;$RB_c$) = ${\cal L}(OCN)$.



\begin{definition}\rm  
Let $G=(G_1, G_2, R)$ be an ($RL$; $RB_c$)-system. The {\it state diagram} ${\cal D}_G=({\cal V}, {\cal E})$ of $G$, 
is a directed graph in which ${\cal V}=\{(X,Y)| X\hspace{-0.1cm}\in ntal(G_1),  
 Y\hspace{-0.1cm}\in ntal(G_2)\}\cup \{F\}$ is the set of nodes. The nodes labeled by $(S_1,S_2)$ and by $F$ are the initial node, and respectively, the final node of ${\cal D}_G$. Let $r=(r_{G_1}, r_{G_2})$ be an arbitrary rule in $R$, the set ${\cal E}$ of edges is built as follows. 

- If $r_{G_1}$ is of the form $X\rightarrow xY$, $x\in tal(G_1) \cup \{\lambda\}$, $X,Y\in ntal(G_1)$,  
and $r_{G_2}$ is of the form $S\rightarrow S$, $S\in \{S_2, Z_2\}$, then $\cal E$ contains an edge, 
labeled by $x/0$, from the node labeled by $(X,S)$ to the node labeled by $(Y,S)$. 

- If $r_{G_1}$ is of the form $X\rightarrow xY$, $x\in tal(G_1) \cup \{\lambda\}$, $X,Y\in ntal(G_1)$ 
and $r_{G_2}$ is of the form $S\rightarrow SZ_2$, $S\in \{S_2, Z_2\}$, then $\cal E$ contains an edge, 
labeled by $x/+$, from the node labeled by $(X,S)$ to the node labeled by $(Y,Z_2)$. 

- If $r_{G_1}$ is of the form $X\rightarrow xY$, $x\in tal(G_1) \cup \{\lambda\}$, $X,Y\in ntal(G_1)$ 
and $r_{G_2}$ is of the form $Z_2\rightarrow \lambda$, then $\cal E$ contains two edges, labeled by $x/-$,   
one from the node labeled by $(X,Z_2)$ to the node labeled by $(Y,Z_2)$, and the other one from $(X,Z_2)$
to the node labeled by $(Y, S_2)$. 

- If $r_{G_1}$ is of the form $X\rightarrow x$, $x\in tal(G_1) \cup \{\lambda\}$, $X\in ntal(G_1)$ and 
$r_{G_2}$ is of the form $S_2\rightarrow \lambda$, then $\cal E$ contains an edge, labeled by $x$, from 
the node labeled by $(X,S_2)$ to a final node labeled by $F$. 
\end{definition}

Next we prove that walking through the state diagram ${\cal D}_G$, from the initial node to the final node, 
with an additional counting procedure for the number of $Z_2$'s existing in the s.form of $G_2$, we can 
simulate the work of $G$ by using (nondeterministic) logarithmic space. More precisely, we have

\begin{theorem}\rm
Each $L\in{\cal L}$($RL$; $RB_c$) can be recognized by a nondeterministic TM in  
${\cal O}(\log n)$ space and ${\cal O}(n^k\log n)$ time (${\cal L}$($RL$; $RB_c$) $\subseteq NSPACE$($\log n$)).  
\end{theorem}

\begin{proof} 
Let $G=(G_1, G_2, R)$ be an ($RL$; $RB_c$)-system with $|ntal(G_1)|\geq 2$, and ${\cal D}_G$ its state 
diagram. Let $\cal T$ be an on-line nondeterministic TM (with stationary positions) composed of an input 
tape that stores an input $w\in (tal(G_1))^*$, $w =x_1x_2...x_n$, of length $n$, and one read-write  
${\cal O}(\log n)$-space restricted working tape used to store in binary the number $n_z$  of $Z_2$'s occurring 
in the s.form of $G_2$ at any snapshot. At the beginning of the computation $n_z$ is set to $0$. Each node 
$(X,Y)\in {\cal V}$ corresponds to a state\footnote{The structure of ${\cal D}_G$ 
is depicted by the manner in which the transition relation $\delta$ of $\cal T$ is defined. Hence, no other 
recordings on the working tape are needed.} of ${\cal T}$, denoted by $q_{(X,Y)}$. Besides, ${\cal T}$ has 
a particular state $q_z$ in which it enters each time ${\cal T}$ reaches a state $q_{(X, S_2)}$ with a 
non-empty counter ($n_z>0$). In $q_z$ the $\delta$ (transition) relation is left undefined.   

Simulating $G$, $\cal T$ starts the computation in $q_{(S_1, S_2)}$ and nondeterministically applies one of the 
rules that rewrites $(S_1, S_2)$, i.e., from $q_{(S_1, S_2)}$ $\cal T$ enters in one of the states $q_{(X, Y)}$ such 
that ${\cal E}$  contains the edge $(S_1, S_2)-(X, Y)$. When passing from $q_{(S_1, S_2)}$ to $q_{(X, Y)}$, 
$\cal T$ reads $x$ from the input tape or reads no symbol if $x=\lambda$, where $x$ is the label of the edge 
$(S_1, S_2)-(X, Y)$. Besides, if $(S_1, S_2)-(X, Y)$ is labeled by  $+$/$-$/$0$ then $n_Z$ is increased by 
$1$/decreased by $1$/left unchanged, respectively. From state $q_{(X, Y)}$, $\cal T$ enters in the state 
$q_{(X', Y')}$ such that $(X, Y)-(X', Y')\in {\cal E}$, it reads $x$ from the input tape, where $x$ is the label 
of the edge $(X, Y)-(X', Y')$, and increases by $1$/decreases by $1$/leaves unchanged the binary value of $n_z$ if 
$(X, Y)-(X', Y')$ is labeled respectively by  $+$/$-$/$0$. The computation continues in this way until all symbols from 
the input tape are read. The input $w$ is accepted if at the end of computation ${\cal T}$ enters in the final state $q_F$ 
with no $Z_2$'s on the s.form of $G_2$ (i.e., $n_z=0$) and the input head placed on a rightmost delimiter of the input (e.g. the blank symbol).

Note that walking on ${\cal D}_G$ one can reach again the node $(S_1,S_2)$. If at this point $n_z>0$ then 
the computation is blocked due to the existence of state $q_z$. If $n_z=0$ (i.e., the bottom of the stack is 
reached) and there exists a rule of the form $S_1\rightarrow x$, $x\in tal(G_1) \cup \{\lambda\}$, then from $q_{(S_1, S_2)}$, $\cal T$ enters 
in the state $q_F$, and it accepts the input only by reading the blank symbol, otherwise the computation is blocked. If 
$S_1\rightarrow x$ does not exist, $\cal T$ may continue with a rule of the form $S_1\rightarrow xX$, by recharging the counter. 
Furthermore, the simulation allows recharging of the counter each time $\cal T$ gets into a state of the form 
$q_{(X, S_2)}$, with $n_z=0$, and there exists a rule of the form  $X\rightarrow xY$ and unread symbols on the 
input tape. If $\cal T$ reaches $q_{(X, S_2)}$ with a nonempty counter ($n_z>0$) then again the computation is 
blocked due to state $q_z$. 

The time used by $\cal T$ to perform the computation is ${\cal O}(n^k\log n)$, where $k\geq 1$ is a constant. This is because when $G_1$ 
contains no chain rules, $n_z$ (recorded in binary) is linearly bounded by the length of the input (as $G_2$ is a CF grammar). The time complexity may increase for the case when the ($RL$; $RB_c$)-system (through $G_1$) contains chain rules. In this situation, if all chain rules are invoked in loops composed of at least one non-chain rule, then at each run of such a loop at least one symbol is read. Hence, the length of computation depends linearly on the length of the input. Certainly, the same situation occurs when chain rules are sparsely applied during a derivation (in the sense that none of them is used during the execution of a loop). 
If all chain rules are invoked in loops composed only of chain rules, then the running of such a loop is overflowing 
for the counter and useless for the computation. If the $Z_2$'s symbols, added through such loops, are deleted by non-chain rules, i.e., meantime the system scans some input symbols, then the running time depends linearly on the length of the input. Otherwise, if the $Z_2$'s symbols are deleted only by chain rules, then this loops are useless for computation and they can be forced to run of a polynomial number 
of times by imposing ${\cal O}(\log n)$ restriction on the working tape (without altering the language). Hence, the number of $Z_2$'s gathered on the s.form of $G_2$, can be at most polynomial. The time complexity is  ${\cal O}(n^k\log n)$, $k\geq 1$, while the space complexity is ${\cal O}(\log n)$. 
\end{proof} 

It is widely known that ${\cal L}(OC)$ is $NSPACE(\log n)$-complete \cite{Sud}. Therefore, the 
inclusion ${\cal L}(OC)\subseteq SPACE(\log n)$ tracks the equality $SPACE(\log n)$ $=NSPACE(\log n)$.   
Both these statements are still open problems in complexity theory, and therefore 
a $SPACE(\log n)$ upper bound for the whole class of ${\cal L}(OC)$ remains a challenging problem.  
Next we prove that an ${\cal N}{\cal C}^1$, hence a $SPACE(\log n)$, upper bound can be obtained 
for languages generated by ($RL$; $RB_c$)-systems for which $ntal(G_1)$ is restricted to only one 
nonterminal. Denote by ($RL_1$; $RB_c$) these particular ($RL$; $RB_c$)-systems 
and by ${\cal L}(OC_1)$ the class of one counter languages generated by OCA's with 
only one control state. We have ${\cal L}$($RL_1$; $RB_c$) = ${\cal L}(OC_1)$.  

An {\it indexing} ATM is an ATM composed of an input tape, $k$ working tapes ($k\geq 1$), and an index tape. The indexing ATM 
has $k+1$ heads corresponding to the working tapes and the index tape. There is no need of a head for the input tape, since the 
indexing ATM reads a symbol from the input tape, through the index tape as follows. It writes an integer $i$, $i\leq n$, in binary, on 
the index tape, where $n$ is the length of the input, and enters in a state requiring the symbol placed at the $i^{th}$ location 
on the input tape to be visible to the machine. From the initial state or a given node in the computation tree of the ATM, an indexing ATM may 
proceed either with an {\it universal configuration} (of which state is an universal state) or with an {\it existential configuration} (of which 
state is an existential state). Since an ATM is a non-deterministic machine, there will be one or more such universal (or existential) configurations 
placed at the same level in the computation tree. All edges going from a certain node to universal (existential) configurations, placed at the 
same level, are called universal (existential) {\it branches}. Sometimes, when in an universal (existential) 
branch the ATM makes some more complex operations (such as simulating another ATM), we call the corresponding branch 
universal (existential)  {\it process}.    

The \textit{Nick class} (${\cal N}{\cal C}$) is defined as the class of all functions computable by a family of \textit{uniform Boolean 
circuits} with polynomial size and depth bounded by a polynomial in log $n$. For each integer $i$, we denote by ${\cal N}{\cal C}^i$ 
the class of functions computable by polynomial size Boolean circuits with depth $\mathcal O$(log$^i$ $n$) and fan-in two. We 
have ${\cal N}{\cal C}^i\subseteq {\cal N}{\cal C}^{i+1}$, $i\geq 1$, and ${\cal N}{\cal C}=\cup_{i\geq 1}{\cal N}{\cal C}^i$. 
Depending on the type of the uniformity restrictions, i.e., speedups on time and space imposed on (alternating) TMs  
simulating a family of circuits, we obtain several ${\cal N}{\cal C}$ ``uniform'' classes. The uniformity condition on which we 
are interested in this paper is the $U_{E^*}$-uniformity. It concerns logarithmic restrictions on the time and space needed by 
an ATM to simulate a family of circuits of size $s(n)$ and depth $t(n)$ \cite{R}. 
Depending on the type of the uniformity, several characterizations of the ${\cal N}{\cal C}$ class in terms of ATM resources are 
presented in \cite{R}. For $i\geq 2$, all ${\cal N}{\cal C}^i$ classes behave similarly, no matter which uniformity
restriction is imposed on circuits, i.e., ${\cal N}{\cal C}^i =ASPACE, TIME(\log n, \log^i n)$,
where $ASPACE, TIME(s(n), t(n))$ denotes the class of languages acceptable by ATMs in simultaneous space $s(n)$ and time $t(n)$. For $i=1$ this equality 
holds only for the $U_{E^*}$-uniformity, more precisely, $U_{E^*}$-uniform ${\cal N}{\cal C}^1 = ASPACE, TIME(\log n, \log n)= ALOGTIME$. Furthermore, 
${\cal N}{\cal C}^1\subseteq SPACE(\log n) \subseteq NSPACE(\log n)$ \cite{BDG}. 
Recall that an ($RL_1$; $RB_c$)-counter system performs a {\it zero-test} ($0$-test) each time the 
$RB_c$-grammar rewrites $S_2$ after some deletions of $Z_2$'s. After a $0$-test the counter can be 
recharged. We have 

\vspace*{0.1cm}
\begin{theorem}\rm
Each language $L\in{\cal L}$($RL_1$; $RB_c$) can be recognized by an indexing ATM in ${\cal O}(\log n)$ 
time and space (i.e., ${\cal L}$($RL_1$; $RB_c$) $\subseteq \hspace*{0.1cm}ALOGTIME$). 
\end{theorem}

\begin{proof}
Let $G=(G_1, G_2, R)$ be an ($RL_1$;$RB_c$)-system, $L(G)$ the language generated by $G$, $tal(G_1)=\{a_1,..., a_m\}$, $ntal(G_1)=\{S_1\}$, and $ntal(G_2)=\{S_2, Z_2\}$. $R=rew(G)$ can be composed of the following rules (where $x\in tal(G_1)$) 
\vspace{0.1cm}\\ 

\hspace{1cm}${\cal S}=\left\{\begin{array}{llllll} \psi^{x}_1=(S_1\rightarrow xS_1, S_2\rightarrow S_2), \psi^{x}_2=(S_1\rightarrow xS_1, S_2\rightarrow S_2Z_2),\\
\psi^{x}_3=(S_1\rightarrow xS_1, Z_2\rightarrow Z_2Z_2), \psi^{x}_4=(S_1\rightarrow xS_1, Z_2\rightarrow Z_2), \\
\psi^{x}_5=(S_1\rightarrow xS_1, Z_2\rightarrow \lambda), \psi^{x}_6=(S_1\rightarrow x, S_2\rightarrow \lambda),\\
\psi_7=(S_1\rightarrow S_1, S_2\rightarrow S_2Z_2), \psi_8=(S_1\rightarrow S_1, Z_2\rightarrow Z_2Z_2),\\ 
\psi_{9}=(S_1\rightarrow S_1, Z_2\rightarrow \lambda),\psi_{10}=(S_1\rightarrow \lambda, S_2\rightarrow \lambda).
\end{array}\right.$
\vspace{0.1cm}\\ 

Let $\cal A$ be an ATM composed of an index tape, a (read-write) working tape divided into three blocks 
${\cal B}_1$, ${\cal B}_2$, and ${\cal B}_3$, and an input tape that stores an input word $w\in (tal(G_1))^*$, $w=x_1...x_n$, of length $n$. The first block ${\cal B}_1$, of constant size (related to the length of the input) 
is used to store the data concerning the types of the rules in $R$, the second block ${\cal B}_2$, of 
${\cal O}(\log n)$ size, is used to record in binary the number $n_z$ of $Z_2$'s existing on the s.form of 
$G_2$ at any snapshot, while the last block ${\cal B}_3$ is used for auxiliary computations. 

In brief, we eliminate in turn some of the rules in $\cal S$ and we study whether the rest of the rules form a 
valid ($RL_1$; $RB_c$)-system. For a valid ($RL_1$; $RB_c$)-system we prove that $w\in (tal(G_1))^*$ can 
be accepted by $\cal A$ in logarithmic time and space. We consider only those ($RL_1$; $RB_c$)-systems that 
may rise interesting computations, depending on the types of the rules that compose $R=rew(G)$. 

Any successful derivation in a ``general'' ($RL_1$; $RB_c$)-system runs upon two main ``subroutines''. 
Rules applied on the first subroutine cannot be applied on the second one. On the first subroutine the 
system applies, arbitrarily many times, a rule of type $\psi^{x}_1$. Applying one time a rule of type  
$\psi^{x}_6$ or $\psi_{10}$ the system ends up with a successful derivation, or it may enter in the second 
subroutine by applying (one time) a rule of type  $\psi^{x}_2$ or  $\psi_7$. Now it can apply, arbitrarily many 
times, rules of the forms $\psi^{x}_3$, $\psi^{x}_4$, $\psi^{x}_5$, $\psi_8$, $\psi_9$. In a successful derivation 
the second subroutine must be canceled  by the application of a rule of type  $\psi^{x}_5$ or  $\psi_9$, to delete 
the last $Z_2$ symbol existing on the s.form of $G_2$ (the counter stack) performing a $0$-test. Then it may finish 
the computation (by using a rule of type $\psi^{x}_6$ or  $\psi_{10}$)  or it may recharge the counter, by passing 
either into the first subroutine or into the second one. 
In the sequel, rules of types $\psi_8$ and $\psi_9$ are called correction rules, because they do not produce terminals, 
but they add more $Z_2$'s on the s.form of $G_2$ (to allow applications of non-chain rules in $G_1$ associated with 
deleting rules in $G_2$) or they delete an exceeding number of $Z_2$'s. 

In stationary positions $\cal A$ records on ${\cal B}_1$ the rules in $R$, while $n_z$ is set to $0$ and recorded on ${\cal B}_2$. In an existential 
state $\cal A$ guesses the length of $w$, i.e., writes on the index tape $n$ in binary, and checks whether the $n^{th}$ cell of the input tape 
contains a symbol in $tal(G_1)$ and the cell $n+1$ contains no symbol. 

Denote by $\wp_i$ the ($RL_1$; $RB_c$)-system composed of some rules described in ${\cal S}$. There are finitely many such systems and most of them are invalid (do not generate a language). Next we discus the solution for the most important ones.   


- $\wp_1$: ($\psi^x_1, \psi^x_6$) is an ($RL_1$; $RB_c$)-system composed of rules of types $\psi^x_1$ and $\psi^x_6$. $\cal A$ checks whether $R$ contains a rule of type $\psi^{x}_6$ that produces $x_n=x$, and whether for each $x_i$ occurring in $x_1x_2...x_{n-1}$, $R$ contains the rule $\psi^{x_i}_1$ that produces $x_i$. These can be done by universally branching all $x_i$, $1\leq i \leq n$. 

- $\wp_2$: ($\psi^x_1, \psi_{10}$) $\cal A$ universally checks\footnote{For a system  ($\psi^x_1, \psi^x_6, \psi_{10}$) the procedure works similar as described in $\wp_1$ and $\wp_2$.} whether for each terminal $x_i$ occurring in $x_1x_2...x_n$ there exists in $R$  a rule of the form $\psi^{x_i}_1$, $1\leq i \leq n$. 

- $\wp_3$: ($\psi_7, \psi^x_3, \psi^x_5, \psi_{10}$)
\textit{The case when each $x$ occurring in $w$ is produced in the same time by $\psi^x_3$ and $\psi^x_5$.} $\cal A$ checks whether \hspace{0.1cm}$I.$  
Each $x_i$ in $w$ is produced (simultaneously) by rules of types $\psi^{x_i}_3$ and $\psi^{x_i}_5$, $1\leq i \leq n$ ($n$ universal branches).

If \textit{$n$ is an odd number} then a strategy that $G$ may choose, 
to derive $w$, is to apply first rules $\psi^{-}_3$ of $\frac{n-1}{2}$ times, and then rules $\psi^{-}_5$ of $\frac{n-1}{2}+1$ times (in order to delete all $Z_2$'s from the s.form of $G_2$). It ends  the derivation 
with rule $\psi_{10}$. 

If \textit{$n$ is even} then a strategy that $G$ may choose is to apply alternatively 
the rules $\psi^x_5$ and $\psi_7$, performing $n$ zero-tests and ending with $\psi_{10}$.  
Hence, any combination of terminals in $tal(G_1)$ that satisfy $I$. is accepted (regular language). 

- $\wp_4$: ($\psi_7, \psi^x_5, \psi_{10}$) The system performs $n$ $0$-tests, similar for ($\psi_7, \psi^x_2, \psi^x_5, \psi_{10}$), while for ($\psi_7, \psi^x_5, \psi_8, \psi_{10}$) one $0$-test suffices ($\cal A$ returns $1$),  ($\psi_7, \psi^x_3, \psi_{10}$) or ($\psi^x_2, \psi^x_3, \psi_{10}$) are invalid systems (too many of $Z_2$'s, $\cal A$ returns $0$).



- $\wp_5$: ($\psi^x_2, \psi^x_5, \psi_{10}$) I. \textit{If each $x$ occurring in $w$ is produced in the same time by $\psi^x_2$ and $\psi^x_5$}, then $\cal A$ performs $n$ $0$-tests. II. \textit{If each $x$ occurring in $w$ is produced either by $\psi^x_2$ or by $\psi^x_5$}, then $\cal A$ checks (through universal branches) whether each $x_i$, where $i$ is odd, is produced by a rule $\psi^{x_i}_2$ and each $x_i$, where $i$ is even, is produced by rule $\psi^{x_i}_5$. III. \textit{If there are $x$ occurring in $w$ produced in the same time by rules of types $\psi^x_2$ and $\psi^x_5$, and there are $x$'s occurring in $w$ produced only by a rule of type $\psi^x_2$ or $\psi^x_5$}, then $\cal A$ checks whether $w$ contains at least one substring composed of consecutive symbols produced only by rules of type 
$\psi^{-}_2$ or of type $\psi^{-}_5$ (but not both). If it holds,  $\cal A$ does not accept the input (it returns $0$, and $1$ otherwise). 

- $\wp_6$: ($\psi_7, \psi^x_3, \psi_9, \psi_{10}$)  Due to rule $\psi_9$, $\cal A$ always accepts as long as each $x_i$ in $w$ is produced by a rule of type $\psi^{x_i}_3$, $1\leq i \leq n$ (universal branches). 

- $\wp_7$: ($\psi_7, \psi^x_3, \psi^x_5, \psi^x_6, \psi_9$) and \textit{each $x$ occurring in $w$ is produced in the same time 
by $\psi^x_3$ and $\psi^x_5$.} $\cal A$ proceeds similar as in $\wp_3$. It also checks if $x_n$ is produced by  $\psi^{x_n}_6$.


- $\wp_8$: ($\psi_7, \psi^x_3, \psi^x_4, \psi^x_5, \psi_{10}$). I. If \textit{$\psi^x_3$ and $\psi^x_5$ produce the same terminals, but distinct than the terminals produced by $\psi^x_4$} or if \textit{$\psi^x_3$, $\psi^x_4$, and $\psi^x_5$ produce the same terminals}, then $\cal A$ proceeds similar as in $\wp_3$, as the rules of type $\psi^x_4$ do not modify the number of $Z_2$'s on the s.form of $G_2$ (there will always be a derivation of $w$ in $G$ by the nondeterminism and $0$-tests). 
II. If \textit{each $x$ in $w$ is produced only by a rule in $\{\psi^x_3, \psi^x_4, \psi^x_5\}$} then  $w$ is of the form  $w=\alpha_1\beta_1\alpha_2\beta_2...\alpha_k\beta_k\alpha_{k+1}$ in which each substring  $\alpha_i$ is composed of terminals produced by rules of type $\psi^x_5$ and additionally of type $\psi^x_4$ (however these rules may be missing), and each substring $\beta_i$ is composed of terminals produced by rules of type $\psi^x_3$, and additionally, of type $\psi^x_4$. Denote by $\eta^{(5)}_i$ the number of terminals in $\alpha_i$ produced by rules of type $\psi^x_5$, and by $\eta^{(3)}_i$ the number of terminals in $\beta_i$ produced by rules of type $\psi^x_3$. Note that $\eta^{(5)}_1$ may be $0$. If $\eta^{(5)}_1>0$ then the ($RL_1$; $RB_c$)-system should perform $\eta^{(5)}_1$ $0$-tests, by alternatively applying a rule $\psi^x_5$, followed by $\psi_7$ (to recharge the counter). 
The same happens each time, at a certain snapshot (point of derivation), the number of terminals produced by rules of type $\psi^x_3$ is smaller than the number of terminals produced by rules of type $\psi^x_5$ (the rule $\psi_7$ allows the counter to be recharged at any $0$\hspace{0.1cm}-test). Hence, for this kind of system, an input $w$ can be decided when checking the $\beta_k$ segment. Note that, if $\eta^{(5)}_{k+1}=0$ then $w$ cannot be generated by $G$, as the counter cannot be emptied. Since a rule of type $\psi^x_4$ does not modify the number of $Z_2$'s on the s.form of $G_2$, $\cal A$ only checks whether terminals in $w$ that are not produced by rules of type $\psi^{-}_3$ and $\psi^{-}_5$, can be produced by rules of type $\psi^{-}_4$. 
Below we describe a sequential procedure that decides $w$.
\vspace{0.1cm}

First a variable $S$ is set to $\eta^{(3)}_1$, $\eta^{(3)}_{k+1}=0$, and $i=2$, \textbf{while}  $i\leq k+1$ \textbf{do} 
$\{$ \hspace{-0.1cm}\textbf{if}  $S+1 \leq \eta^{(5)}_i$ \textbf{then} $S$ is set to $\eta^{(3)}_i$ (because an $S+1$ number of $Z_2$'s are deleted by rules of type $\psi^x_5$ and $\eta^{(5)}_i-1-S$ $0$-tests are performed, hence $S$ is set to the next $\eta^{(3)}_i$) \textbf{else} $S$ is set to $S - \eta^{(5)}_i+\eta^{(3)}_i$ (the counter contains $S - \eta^{(5)}_i +\eta^{(3)}_i+1$ of $Z_2$'s),  $i$ is set to $i+1$ \textbf{$\}$}. The procedure returns $S$. 
\vspace{0.1cm}

The input $w$ can be derived by $G$ if $S=0$, as $S$ is set to $\eta^{(3)}_{k+1}=0$, only if the previous value of $S$ is smaller than $\eta^{(5)}_{k+1}-1$. 
For $\cal A$ to decide the input only the part of the routine starting from the point when $i$ is set to the maximum $j$, $j\leq k$, such that $S$ is set to $\eta^{(3)}_j$, after the last $0$-test, is important (no matter how many times the counter is recharged).       
To do so,  $\cal A$ spawns $k$ existential processes $\partial_j$, $1\leq j \leq k-1$, each of which holds a suffix of $w$ of the form 
$w_{j,k+1}=\beta_j\alpha_j\beta_{j+1}...\alpha_k\beta_k\alpha_{k+1}$. On each $\partial_j$, $\cal A$
checks whether $\sum^{k}_{i=j} \eta^{(3)}_i -\sum^{k}_{i=j+1} \eta^{(5)}_i +1 \leq  \eta^{(5)}_{k+1}$, i.e., whether the 
number of $Z_2$'s gathered on the counter, over the $\beta_l$ segments,  $j\leq l \leq k+1$, can be deleted by rules 
of type  $\psi^{-}_5$, used over the $\alpha_l$ segments, as long as no other $0$\hspace{0.1cm}-test is performed during the scanning of the suffix $w_{j,k+1}$. That is, whether  $\sum^{q}_{i=j} \eta^{(3)}_i -\sum^{q}_{i=j+1} \eta^{(5)}_i +1 > \eta^{(5)}_{q+1}$, for each $j\leq q \leq k-1$. These inequalities should be universally checked on each existential branch, through $O(k-j+1)$ universal processes, each of which testing only one inequality.       
The fan-out of the computation tree of $\cal A$ on the $k$ existential processes $\partial_j$, or on the 
$k-j+1$ universal branches $1\leq j \leq k-1$, is $O(n)$. It can be converted (using a divide and conquer procedure) into a computation tree of fan-out $2$ and height $O(\log n)$. The (iterative) summations and comparison operations in binary require $O(\log n)$ time and space, as these operations\footnote{Multiplication and addition of $n$ $n$-bit numbers are in ${\cal N}{\cal C}^1$.} are in ${\cal N}{\cal C}^1$. Hence, $\cal A$  decides $w$ in $O(\log n)$ time and space. 

\textit{When the correction rule $\psi_9$ is added to $\wp_8$} $\cal A$ accepts any combination of 
$a_i\in tal(G_1)$ in $w$ as the surplus of $Z_2$'s on the counter can be deleted by $\psi_9$ at the end of 
computation. In this case $\cal A$ universally checks whether each $x_i$ 
in $w=x_1x_2...x_n$ is produced solely by a single rule in $\wp_8$.  Other cases, e.g., when $\psi^{x}_6$ is added 
to $\wp_8$, and $\psi^{x}_3$, $\psi^{x}_4$, or $\psi^{x}_5$ are dropped from $\wp_8$, such that each $x$ is 
produced only by a rule in $\wp_8$, are treated in a similar manner. 

- $\wp_9$: ($\psi_7, \psi^x_3, \psi^x_5, \psi_{10}$) and \textit{each $x$ in $w$ is produced only by a rule in $\{\psi^x_3, \psi^x_5\}$} then $\cal A$ proceeds similar as in $\wp_8$ II. 

- $\wp_{10}$: ($\psi_7, \psi^x_3, \psi^x_5, \psi_{10}$) and \textit{there exist $x$'s in $w$  produced only by a rule in $\{\psi^x_3, \psi^x_5\}$ and there exist $x$'s produced in the same time by $\psi^x_3$ and  $\psi^x_5$.} 

Denote by $I=\{a_{i_1}, ...,a_{i_r}\}$, $H=\{a_{h_1}, ...,a_{h_p}\}$ and $L=\{a_{l_1}, ...,a_{l_q}\}$ the 
sets of terminals in $tal(G_1)$ such that each $a_{i_j}$, $1\leq j \leq r$, is produced only by the rule $\psi^{a_{i_j}}_3$, each $a_{h_j}$, $1\leq j \leq p$, is produced only by $\psi^{a_{h_j}}_5$, and each 
$a_{l_j}$, $1\leq j \leq q$, is produced in the same time by $\psi^{a_{l_j}}_3$ and $\psi^{a_{l_j}}_5$. 
Obviously, $H \cup I \cup L$ is a partition of $tal(G_1)$. 
In this case an input $w$ may be of the form\footnote{Some of the substrings of type $H_{-}$, $I_{-}$, or $L_{-}$ may be empty strings.}  $w=L_0H_0L_1\alpha_1...\alpha_j...\alpha_k\bar{L}_0\bar{H}_0\bar{L}_1$, where 
$\alpha_j=I_jL_{j_1}H_{j_1}L_{j_2}H_{j_2}...L_{j_q}H_{j_q}$, $1\hspace{-0.1cm}\leq j \leq k$, each 
$H_{-}$ is composed of terminals in $H$, each $I_{-}$ is composed of terminals in $I$,  and each $L_{-}$ is composed of terminals in $L$. Substrings of type $H_{-}$ may be resolved by $0$-tests and rule $\psi_7$, 
substrings of type $L_{-}$ may be resolved as in $\wp_{3}$.  The only substrings that may block the computation are of type $I_{-}$.


To check whether the $Z_2$'s symbols added on the counter by segments of type $I_{-}$ can be ``deleted in due time'', as in the case of systems of type $\wp_{8}$, $\cal A$ spawns $k$ existential processes $\partial_j$, 
$1\leq j \leq k$, each of which holds a suffix of $w$ of the form\footnote{The beginning of each $I_{-}$ segment may be considered as a marker in $w$.} $w_{j,k}=\alpha_j\alpha_{j+1}...\alpha_k\bar{L}_0\bar{H}_0\bar{L}_1$, $1\leq j \leq k-1$. $\cal A$ checks for $w_{j,k}$ similar conditions as in $\wp_{8}$, where the number of terminals produced by rules of type $\psi^{-}_{5}$ are counted, then when it is necessary to delete the exceeding $Z_2$'s produced by rules of type $\psi^{-}_3$ used in $I_{-}$, also over the segments of type $L_{-}$.     

If the number $\eta^{(3)}_l$ of $Z_2$'s produced by rules of type $\psi^{-}_3$ to generate a certain $I_l$ segment is smaller than the number $\eta^{(5)}_l$ of $Z_2$'s deleted by rules of type $\psi^{-}_{5}$ (used in $H_{l_1}$,..., $H_{l_q}$ and $L_{-}$ segments), then denote by $S^-_l= \eta^{(5)}_l-\eta^{(3)}_l$ (computed in binary). $\cal A$ performs $S^-_l$ $0$-tests for some of the $H_{l_1}$, ..., $H_{l_q}$ segments and applies the procedure described in $\wp_{3}$ for the $L_{-}$ segments, or it may use some of the $\psi^{-}_{5}$ rules to delete the surplus of $Z_2$'s produced in segments of type  $\alpha_{-}$ occurring in $w_{j,k}$ before $\alpha_l$ (if this surplus exists). 

If the  number $\eta^{(3)}_l$ of $Z_2$'s produced by rules of type $\psi^{-}_3$ to generate $I_l$ is greater than the 
number $\eta^{(5)}_l$ of $Z_2$'s deleted by rules of type $\psi^{-}_{5}$ used to generate the $H_{l_1}$,..., $H_{l_q}$ and $L_{-}$ segments, then denote by $S^+_l= \eta^{(3)}_l-\eta^{(5)}_l$ the surplus of $Z_2$'s (computed in binary). $\cal A$ checks further for the next segments of type  $\alpha_{l'}$ occurring in $w_{j,k}$, after $\alpha_l$, such that $\eta^{(3)}_{l'} < \eta^{(5)}_{l'}$. Denote by $S^-_{l'}=\eta^{(5)}_{l'} - \eta^{(3)}_{l'}$. $\cal A$ computes  $S^+_l - S^-_{l'}$, if  $S^+_l > S^-_{l'}$,  $l<l'$, and $S^+_l$ is set to $S^+_l - S^-_{l'}$. Otherwise $S^+_l$ is set to $0$ and $\cal A$ performs $S^-_{l'} -S^+_l$ $0$-tests on the $\alpha_{l'}$ segment. 


- $\wp_{11}$: ($\psi^x_2, \psi^x_3, \psi^x_5, \psi_{10}$) \textit{There exist $x$'s in $w$  produced only by a rule in 
$\{\psi^x_3, \psi^x_5\}$ and there exist $x$'s produced in the same time by $\psi^x_3$ and  $\psi^x_5$.}

$\wp_{11}$ $I.$ Any $x$ produced by a rule of type $\psi^{x}_2$ can be also produced by a rule of type $\psi^{x}_3$ or $\psi^{x}_5$, or both. 
Denote by $I=\{a_{i_1}, ...,a_{i_r}\}$, $H=\{a_{h_1}, ...,a_{h_p}\}$, and $L=\{a_{l_1}, ...,a_{l_q}\}$ the sets 
of terminals in $tal(G_1)$ such that each $a_{i_j}$, $1\leq j \leq r$, is produced only by the rule $\psi^{a_{i_j}}_3$, 
each $a_{h_j}$, $1\leq j \leq p$, is produced only by $\psi^{a_{h_j}}_5$, and each $a_{l_j}$, $1\leq j \leq q$, 
is produced in the same time by $\psi^{a_{l_j}}_3$ and $\psi^{a_{l_j}}_5$. For the sake of simplicity, we refer to a terminal 
$x\in X$, $X\in \{I,H,L\}$, that occurs in $w$, as $x\in X\cap w$. $\cal A$ spawns $n$ universal processes $\partial_j$, $1\leq j \leq n$. 

 $\diamond$ On $\partial_1$, $\cal A$ returns $1$ if there exists a rule $\psi^{x_1}_2$ producing $x_1$ and $0$ otherwise.  

 $\diamond$ On any $\partial_j$, $2\leq j \leq n-1$, such that $x_j\in I\cup L$, $\cal A$ returns $1$ by default. 

 $\diamond$ On any $\partial_j$, $2\leq j \leq n-1$, such that $x_j\in H$, $\cal A$ counts the number of terminals $x\in H\cup I$ occurring in $w_j=x_1x_2...x_{j-1}$. Suppose there exist $\eta^{(3)}_j$ terminals in $I\cap w_j$ and $\eta^{(5)}_j$ terminals in $H\cap w_j$. $\cal A$ checks whether $\eta^{(3)}_j-\eta^{(5)}_j\geq 0$, i.e., at the application of rule $\psi^{x_j}_5$ (producing $x_j$) the s.form of $G_2$ contains at least one $Z_2$. If it holds $\partial_j$ returns $1$. If $\eta^{(5)}_j\hspace*{-0.1cm}-\eta^{(3)}_j > 1$, $\cal A$ must use a rule of type $\psi^{-}_3$ 
to produce a terminal in $L\cap w_j$. Furthermore, $\cal A$ makes such an election for each process $\partial_m$, $m<j$, handling a terminal $x_m\in H\cap w_j$, with the property that $\eta^{(5)}_m-\eta^{(3)}_m > 1$. 
Hence, at process $\partial_j$, when choosing an $y \in L\cap w_j$ to be rewritten by $\psi^y_3$, $\cal A$ must assure that the number of terminals in 
$L\cap w_j$ is greater than the number of terminals handled by processes of type $\partial_m$, with $\eta^{(5)}_m-\eta^{(3)}_m > 1$\footnote{That is there are enough terminals in $L\cap w_j$ that can be produced by rules of type $\psi^{-}_3$  in order to cover the lack of $Z_2$'s that must be deleted by rules of type $\psi^{-}_5$ that produce only terminals in $H$, up to the $x_j$ symbol.}. 
To do so, $\cal A$ counts the number of terminals $x\in H\cap w_j$, with the property that for each fixed position $x_l$ of $x$ in $w$ (i.e., $x_l=x$) 
$\eta^{(5)}_l-\eta^{(3)}_l > 1$ holds. Suppose there exist $\eta^{(c,j)}_L$ terminals $x$ with this property (including $x_j$), i.e., up to the $j^{th}$ symbol in $w$, $\cal A$ must choose $\eta^{(c,j)}_L-1$ terminals in $L$ to be produced by rules of type $\psi^{-}_3$. $\cal A$ checks whether $\eta^{(j)}_L\geq \eta^{(c,j)}_L$, where $\eta^{(j)}_L$  is the number of terminals in $L\cap w_j$. On the other hand, each time $\eta^{(5)}_j-\eta^{(3)}_j -1 = 0$, a $0$-test is performed, hence A must
check whether $x_j$ can be produced by a rule of type $\psi^{x_j}_2$ (no matter whether $x_j$ can be also produced by a rule of type $\psi^{-}_3$ or $\psi^{-}_5$, or both. Process $\partial_j$ returns $1$ if these conditions hold.
 
 $\diamond$ On $\partial_n$ If $x_n\in I$, then $\partial_n$ returns $0$. 

{\it Case $\partial^1_n$} If $x_n\in H$ and $\eta^{(5)}_n-\eta^{(3)}_n=1$, then after applying $\psi^{x_n}_5$ 
there will be no $Z_2$'s on the s.form of $G_2$, and $G$ may apply rule $\psi_{10}$ to empty the stack. However, this 
is conditioned by the acceptance of the string composed of all terminals in $w\cap L$ that have not been chosen to be produced by rules of type $\psi^{-}_3$. Denote by $w_L$ the string composed of these terminals (no matter which the order is). The acceptance of $w_L$ is decided as in $\wp_3$. We have $|w_L|=\eta^{(n)}_L-\eta^{(c,n)}_L$, where $\eta^{(n)}_L$ is the number of terminals in $L\cap w$, and $\eta^{(c,n)}_L$ is the 
number of terminals in $L\cap w$ chosen to be produced by rules of type $\psi^{-}_3$. Note that, $\eta^{(c,n)}_L=\eta^{(c,k)}_L$, where $x_k$ is the rightmost  terminal in $H\cap w$ with the property $\eta^{(5)}_k-\eta^{(3)}_k > 1$. 

{\it Case $\partial^2_n$} If $x_n\in H$ and $\eta^{(5)}_n-\eta^{(3)}_n> 1$, $\cal A$ proceeds as in process $\partial_j$ 
($x_j\in H$, $\eta^{(5)}_j-\eta^{(3)}_j> 1$) to check whether there exist enough terminals in $L\cap w$ to be produced by rules 
$\psi^{-}_3$, and then as in $\partial^1_n$ to check whether $w_L$ is accepted by $\cal A$. 

{\it Case $\partial^3_n$} If $x_n\in H$ and $\eta^{(3)}_n\hspace*{-0.1cm}-\eta^{(5)}_n \geq 0$, then the exceeding $Z_2$'s existing in the s.form of $G_2$, after ``spelling'' the string composed of terminals in $(H\cup I)\cap w$ 
(preserving their order in $w$) must be deleted by using rules of type $\psi^{-}_5$ that produce terminals in $L\cap w$. To do so, $\cal A$ checks whether  
$\eta^{(n)}_L-\eta^{(c,n)}_L-1 \geq \eta^{(3)}_n-\eta^{(5)}_n$ and the string $w_L$, of length $|w_L|=\eta^{(n)}_L-\eta^{(c,n)}_L-1- (\eta^{(3)}_n-\eta^{(5)}_n)$ is valid as in process $\wp_3$, where $\eta^{(n)}_L$ is the number of terminals in $L\cap w$, and $\eta^{(c,n)}_L$ is the 
number of terminals in $L\cap w$ chosen to be produced by rules of type $\psi^{-}_3$. 

{\it Case $\partial^4_n$} If $x_n\in L$, then denote by $\eta^{(3)}_n$ the number of terminals in $I\cap w$, 
$\eta^{(5)}_n$ the number of terminals in $H\cap w$. $\cal A$ computes $|w_L|=\eta^{(n)}_L-\eta^{(c,n)}_L - 1-(\eta^{(3)}_n-\eta^{(5)}_n)$ and it checks whether $w_L$ is accepted by $\cal A$ (as in $\wp_3$). Recall that $G$ must apply rule $\psi^{x_n}_5$ to produce $x_n$. It ends the derivation with rule $\psi_{10}$. 
    
Note that, the reasoning for which a process $\partial_j$, that handles a terminal $x_j\in I\cup L$ always returns $1$, 
is the following. If $x_j$ is not preceded in $w$ by any $y\in H\cap w$, then applying $\psi^{x_j}_3$ to produce $x_j$ 
arises no problem ($\partial_j$ may return $1$). If $x_j$ is preceded in $w$ by some $y\in H\cap w$, then the application of $\psi^{x_j}_3$ may be disabled by the lack of $Z_2$'s in the s.form of $G_2$. However, this is resolved when checking if $\eta^{(3)}_l-\eta^{(5)}_l\geq 0$, where $y=x_l$ is the very last terminal that precedes $x_j$, occurring also in $H$. If $\eta^{(3)}_l-\eta^{(5)}_l < 0$, then at process $\partial_l$, $\cal A$ chooses a terminal in $L\cap w$ to be produced by a rule of type $\psi^{-}_3$, allowing the application of rule $\psi^{x_j}_3$ (and again $\partial_j$ may return $1$). If $\partial_l$ cannot find enough terminals in $L\cap w$ to be produced by 
a rule of type $\psi^{-}_3$, then $\partial_l$ returns $0$ \footnote{The correction is done when running, through another universal branch, the process $\partial_l$, case when $x_l \in H$ (already described).} which rectifies the possible error obtained when $\partial_j$ is set to $1$ ($\partial_{\_}$ are universal processes). 

$\wp_{11}$ $II.$ If $x$ is produced only by a rule of type $\psi^{x}_2$, then $x$ can be used as a $0$-test marker in $w$ and $\cal A$ checks whether the strings between these markers can be generated by a system composed only of rules of types  $\psi^{-}_3$ and $\psi^{-}_5$.


Other cases, e.g., when $\psi^{x}_6$ or $\psi_9$ are added to $\wp_8$, $\wp_9$, $\wp_{10}$ or $\wp_{11}$ and $\psi_{10}$ is dropped, are treated in a similar manner. 
If in all the above cases $\psi^{x}_1$ is added, and $x$ is produced only by $\psi^{x}_1$, then  again $x$ may be used as a marker in $w$ and $\cal A$ checks whether the substrings of $w$ placed between these markers can be generated by an ($RL_1$; $RB_c$)-system discussed above. 
If $x$ is produced in the same time by $\psi^{x}_1$,  $\psi^{x}_3$ or $\psi^{x}_5$ etc., then $\psi^{x}_1$ may be left to be used only for those $x_i$ in $w$ when $0$-tests may be produced by reading $x_i$, $x_{i+1}=x$ and then when using $\psi^{x}_1$ prevents too many of $Z_2$'s to be added on the counter, if instead the rule $\psi^{x}_3$ would be used, or in order to prevent less $Z_2$'s to be deleted, if instead the rule $\psi^{x}_5$ would be used. Otherwise, if $\psi^{x}_1$ has no ``corrections'' to perform on the number of $Z_2$'s in the s.form of $G_2$, then rule $\psi^{x}_1$ is simply ignored, and the system is judged as in the cases already discussed.    

The computation tree of $\cal A$ has a constant number of (existential or universal) levels (related to the length of the input) of fan-out ${\cal O}(n)$. Each level, in which $\cal A$ spawns ${\cal O}(n)$ existential or universal branches, can be unfolded (through a divide and conquer procedure) into a binary tree of height ${\cal O}(\log n)$. Auxiliary operations such as counting, addition, multiplication of binary numbers, require ${\cal O}(\log n)$ time and space (as these are in ${\cal N}{\cal C}^1$). Hence, the computation tree of $\cal A$ has depth ${\cal O}(\log n)$ and bounded fan-out. 

${\cal O}(\log n)$ space restriction on ${\cal B}_2$ forces $\psi_7$, $\psi_8$ or $\psi_9$ to be applied of at most polynomial number of times (related to the length of the input) without altering the language. 
If $\psi_7$ or  $\psi_8$ are applied in $G$ of an exponential number of times, then if the surplus of $Z_2$'s is deleted only by $\psi_9$ through at least one loop that does not contain any non-chain rule, ${\cal A}$ cuts off the (sequential) running time (i.e. space for the ATM) of these (unproductive) loops. Otherwise, the number of $Z_2$'s on the s.form of $G_2$ is linearly related to the length of the input, hence $n_z={\cal O}(n)$,     
and $O(\log n)$ space suffices to record in binary the number of $Z_2$'s on the s.form of $G_2$ at any snapshot.   
\end{proof}



\begin{corollary}\rm
${\cal L}(OC_1) ={\cal L}$($RL_1$; $RB_c$) $\subseteq {\cal N}{\cal C}^1 \subseteq SPACE(\log n)$.  
\end{corollary}

Denote by ($RL^0_1$; $RB_c$) the ($RL^0$; $RB_c$)-counter nets systems for which $ntal(G_1)$ is restricted to only one nonterminal, i.e., ($RL_1$; $RB_c$)-counter systems that are not allowed to perform $0$-tests, and by OCN$_1$ one counter nets automata with only one control state. 
($RL^0_1$; $RB_c$)-systems may be composed only of rules of the forms $\wp^x_3$, $\wp^x_4$, $\wp^x_5$, $\wp_8$, $\wp_9$.  Hence we have 

\begin{corollary}\rm
${\cal L}(OCN_1) ={\cal L}$($RL^0_1$; $RB_c$) $\subseteq {\cal N}{\cal C}^1$.  
\end{corollary}

\section{Useful Examples} 

\begin{example}\rm
The ($RL^0$; $RB_c$)-counter nets system for which  $tal(G_1)=\{a,b\}$, $ntal(G_1)=\{S_1, X\}$, $al(G_2)=\{Z_2\}$, $ax(G_1)= S_1$,  $ax(G_2)= Z_2$,   and 
$R = \{(S_1\rightarrow aS_1, Z_2\rightarrow Z_2Z_2), (S_1\rightarrow X, Z_2\rightarrow \lambda), (X\rightarrow bX, Z_2\rightarrow \lambda), 
(X\rightarrow b, Z_2\rightarrow \lambda)\}$ generates the OCN language $L = \{a^nb^n |n\geq 1\}$, $L \notin {\cal L}(OCN_1)$.
\end{example}

\begin{example}\rm
The ($RL$; $RB_c$) system with  $al(G_1)=\{a,b\}$, $ntal(G_1)=\{S_1, X\}$, $al(G_2)=\{S_2, Z_2\}$, $ax(G_1)= S_1$, $ax(G_2)= S_2$ and 
$R = \{(S_1\rightarrow S_1, S_2\rightarrow S_2Z_2), (S_1\rightarrow aS_1, Z_2\rightarrow Z_2Z_2), (S_1\rightarrow X, Z_2\rightarrow \lambda), 
(X\rightarrow bX, Z_2\rightarrow \lambda), (X\rightarrow \lambda, S_2\rightarrow \lambda)\}$ generates the OCA language $L = \{a^nb^n |n\in \mathbb{N}\}$, $L \notin $ ${\cal L}(OC_1)$. 
\end{example}

\begin{example}\rm
The ($RL^0_1$; $RB_c$) system for which  $tal(G_1)=\{a,b\}$, $ntal(G_1)=\{S_1\}$, $al(G_2)=\{Z_2\}$,  and 
$R = \{(S_1\rightarrow aS_1, Z_2\rightarrow Z_2Z_2), (S_1\rightarrow bS_1, Z_2\rightarrow \lambda), (S_1\rightarrow \lambda, Z_2\rightarrow \lambda)\}$ 
generates the OCN$_1$ language $L = \{a^{n_1}b^{m_1}a^{n_2}b^{m_2}...a^{n_{k-1}}b^{m_{k-1}}a^{n_k}b^{m_k} | \sum_{i=1}^j n_i \geq \sum_{i=1}^j m_i, \\
1\leq j \leq k-1, \sum_{i=1}^k n_i = \sum_{i=1}^k m_i, k\geq 1, n_i, m_i\in \mathbb{N},  n_i \geq 1, 1\leq i\leq k \}$. 

Note that, if rules $(S_1\rightarrow aS_1, S_2\rightarrow S_2Z_2)$ and $(S_1\rightarrow  \lambda, S_2\rightarrow  \lambda)$ are added to $R$, and rule 
$(S_1\rightarrow \lambda, Z_2\rightarrow \lambda)$ is removed from $R$, then the rezulted system is an ($RL_1$; $RB_c$)-counter  system generating the same language.
\end{example}

\section{Conclusions}

Cts-systems were introduced in \cite{R1} as a unifying framework for several rewriting systems and automata. 
The main advantage of a cts-system is that it splits the simulated device into several subcomponents, each of which having to accomplish a precise task, such that all tasks together carry out the system's work. This modularization facilitates  deeper investigations into the nature of the simulated device, from a generative power and computational complexity perspective. In this paper, we have materialized this idea for Petri nets and counter machines. 

We proved, via cts-systems, that the class of languages generated by $\lambda$-free labeled marked Petri nets, known as ${\cal L}$-languages \cite{H}, is included in 
$NSPACE(\log n)$, while the class of languages generated by one-counter machines with only one control state 
(with or without $0$-tests) is included in $U_{E^*}$-uniform $\cal N$$\cal C$$^1$.  All upper bounds are tight and they bring new insights into the computational complexity of several subclasses of one counter languages, languages generated by VASS in fix dimension, and  Petri nets languages. 

One of our future aims is to extend this line of research upon the complexity of one counter (nets) languages, in order to tighten the known upper bound  $NSPACE(\log n)$ to $\cal N$$\cal C$$^1$. Even if sophisticated, we hope to combine methods in Theorem 4.4  with the state diagram used in Theorem 4.3 to reach a $SPACE(\log n)$ upper bound for  ${\cal L}(OC)$ or ${\cal L}(OCN)$. The method described in Theorem 3.2 combined with a proper state diagram defined for the general case of Petri nets similar to that provided by Definition 4.2, may lead to some improvements on the non-elementary upper bound of the ${\cal L}^{\lambda}_0$ class, as ${\cal L}^{\lambda}_0 =$  ${\cal L}$($RL$; $0S$) \cite{AR}. 

\renewcommand\refname{References}


\begin{thebibliography}{99}


\bibitem{AR}
Aalbersberg, I.J.J., Rozenberg, G.: CTS Systems and Petri Nets,  {\it Theoretical Computer Science} 40, 149--162 (1985).

\bibitem{BDG}
Balc\'{a}zar, J.L., D\'{\i}az, J., Gabarr\'{o}, J.: {\it Structural Complexity}, Springer-Verlag (1990). 



\bibitem{C}
Czerwi\'{n}ski, W., Lasota, S., Lazi\'{c}, R., Leroux, J., Mazowiecki, F.:  
The Reachability Problem for Petri Nets is Not Elementary, STOC 2019, 24--33 (2019). 


\bibitem{EHR1}
Ehrenfeucht, A., Hoogeboom, H.J., Rozenberg, G.: Computations in Coordinated Pair Systems, Technical Report CU-CS-260-83, Univ. of Colorado at Boulder (1984).            

\bibitem{EHR2}
Ehrenfeucht, A., Hoogeboom, H.J., Rozenberg, G.: On Coordinated Rewriting, {\it Fundamentals of Computation Theory} (1985). 

\bibitem{EHR3}
Ehrenfeucht, A., Hoogeboom, H.J., Rozenberg, G.: On the Active and Full Use of Memory in 
Right-Boundary Grammars and Push-Down Automata, {\it Theoretical Computer Science} 48(3): 201--228 (1986).

\bibitem{ELP}
Englert, M., Lazi\'{c}, R., Totzke, P.: Reachability in Two-Dimensional Unary Vector Addition
Systems with States is NL-complete. LICS 2016, 477-–484. ACM (2016).


\bibitem{H}
Hack, M.: Petri Net Languages, Technical Report 159, Lab. for Computation Science, MIT Cambridge, MA (1976). 
\bibitem{Ha}
Harrison, M.A.:  {\it Introduction to Formal Language Theory}. Addison-Wesley, (1978).

\bibitem{HK}
Haase, C.,  Kreutzer, S., Ouaknine, J., Worrell, J.: Reachability in Succinct and Parametric One-Counter
Automata. CONCUR 2009, Concurrency Theory, LNCS 5710, 369--383, Springer (2009).

\bibitem{Ho}
Hopcroft, J.E., Pansiot, J.J.: On the reachability problem for 5-dimensional vector
addition systems. {\it Theoretical Computer Science} 8, 135–-159 (1979). 


\bibitem{J}
Jantzen, M.: On the Hierarchy of Petri Net Languages. {\it RAIRO - Theoretical Informatics and 
Applications - Informatique Théorique et Application} 13.1: 19--30 (1979).



\bibitem{KR}
Kleijn, H.C.M., Rozenberg, G.: On the Role of Selectors in Selective Substitution Grammars, 
{\it Fundamentals of Computation Theory}, 190--198 (1981).

\bibitem{LS1}
Leroux, J., Schmitz, S.: Reachability in Vector Addition Systems is Primitive-Recursive in Fixed Dimension. 
LICS 2019, 1--13. ACM (2019)


\bibitem{L}
Lipton, R.: The Reachability Problem Requires Exponential Space, {\it Technical Report
62}, Yale University, Dept. of CS. (1976).

\bibitem{P}
Petri, C.A.: Kommunikation mit Automaten. Bonn: Institut f$\ddot{u}$r Instrumentelle Mathematik, 
Schriften des IIM nr. 2 (1962).

\bibitem{R1}
Rozenberg, G.: Selective Substitution Grammars (Towards a Framework for Rewriting Systems), Part I:  
Definitions and Examples, {\it Elektron. Informationsverarbeit Kybernetik} 13, 455--463 (1977).

\bibitem{R2}
Rozenberg, G.: On Coordinated Selective Substitutions: Towards a Unified Theory of Grammars and Machines,  
{\it Theoretical Computer Science} 37, 31--50 (1985).

\bibitem{R}
Ruzzo, W.: On Uniform Circuit Complexity, {\it Journal of Computer and System Sciences},  22(3), 365--383 (1981).

\bibitem{Sa}
Salomaa, A.: {\it Formal Languages}. Academic Press, London-New York (1973). 

\bibitem{S}
Schmitz, S.: Complexity Hierarchies Beyond ELEMENTARY. {\it ACM Transactions and Computation Theory} 8(1), 3:1--3:36 (2016).

\bibitem{Sud}
Sudborough, I.H.: On Tape-Bounded Complexity Classes and Multihead Finite Automata, 
{\it Journal of Computer and System Sciences}, 10(1): 62--76 (1975).

\end{thebibliography}
\end{document}